\documentclass[a4paper]{article}
\usepackage[margin=1.12in]{geometry}

\usepackage{microtype}%

\usepackage{amsmath}
\usepackage{amssymb}
\usepackage{amsthm}
\theoremstyle{definition}

\newtheorem{theorem}{Theorem}[section]

\newtheorem{lemma}[theorem]{Lemma}

\usepackage{enumerate}

\usepackage{algorithm}
\usepackage{algpseudocode} %
\algrenewcommand\textproc{\textsf}
\algnotext{EndProcedure}
\algnotext{EndIf}

\usepackage{alltt}
\usepackage{graphicx}
\usepackage{url}
\usepackage{latexsym}

\usepackage{hyperref}

\newcommand{\complexf}[1]{\mathcal{#1}}   %
\newcommand{\cplC}{\complexf{C}}   %
\newcommand{\cplD}{\complexf{D}}   %
\newcommand{\cplI}{\complexf{I}}   %
\newcommand{\cplO}{\complexf{O}}   %

\newcommand{\keyword}[1]{\emph{#1}}

\newcommand{\WScan}{\textproc{WScan}}
\newcommand{\WOScan}{\textproc{WOScan}}
\newcommand{\ISnap}{\textproc{IS}}
\newcommand{\IISnap}{\textproc{IIS}}
\newcommand{\ISnapO}{\textproc{IS'}}

\newcommand{\skel}[1]{\operatorname{skel}^{(#1)}}

\newcommand{\Star}{\operatorname{St}}

\newcommand{\join}{\mathbin{\ast}}

\newcommand{\cls}[1]{\overline{#1}}
\newcommand{\Div}[1]{{\operatorname{#1}}}

\newcommand{\Chromatic}{\Div{Ch}}
\newcommand{\Schlg}{\Div{Sch}}
\newcommand{\Sch}[2]{\Schlg^{#1}_{#2}}
\newcommand{\Carrier}{\operatorname{Carr}}
\newcommand{\pmap}[1]{\pi_{#1}}

\newcommand{\coloring}{\operatorname{color}}

\newcommand{\abs}[1]{{\mid}#1{\mid}}

\title{Schlegel Diagram and Optimizable Immediate Snapshot Protocol\thanks{%
This work is to appear in
\href{http://opodis2017.campus.ciencias.ulisboa.pt/}{%
\emph{OPODIS 2017 --- The 21st International Conference on Principles of Distributed Systems.}}}}

\author{Susumu Nishimura\\
Dept. of Mathematics,
   Graduate School of Science, Kyoto University\\ %
 \texttt{susumu@math.kyoto-u.ac.jp}}
\date{}

\begin{document}
\maketitle

\begin{abstract}
In the topological study of distributed systems,
the immediate snapshot is the fundamental computation block
for the topological characterization of wait-free solvable tasks.  %
However, in reality, the immediate snapshot is not available
as a native built-in operation on shared memory distributed systems.
Borowsky and Gafni have proposed
a wait-free multi-round protocol that implements the immediate snapshot using
more primitive operations, namely the atomic reads and writes.

In this paper,
up to an appropriate reformulation on the original protocol
by Borowsky and Gafni,
we establish a tight link between each round of the protocol
and a topological operation of subdivision using Schlegel diagram.
Due to the fact shown by Kozlov that
the standard chromatic subdivision is obtained by iterated
subdivision using Schlegel diagram,
the reformulated version is proven to compute the immediate snapshot
in a topologically smoother way.
We also show that the reformulated protocol is amenable to
optimization:
Since each round restricts the possible candidates of output
to an iteratively smaller region of finer subdivision,
each process executing the protocol can decide at an earlier round,
beyond which the same final output is reached no matter how
the remaining rounds are executed.
This reduces the number of read and write operations involved
in the overall execution of the protocol,
relieving the bottleneck of access to shared memory.
\end{abstract}

{
\interfootnotelinepenalty=10000

\renewcommand{\topfraction}{0.9}
\renewcommand{\bottomfraction}{0.9}
\renewcommand{\floatpagefraction}{0.9}

\section{Introduction}
\label{sec:intro}

The snapshot models
\cite{AttiyaRajsbaum02,BorowskyGafni:STOC93,BorowskyGafni97,Herlihy:DCTopology}
for the shared memory distributed system have been intensively studied
for the analysis of solvability of distributed tasks in the
wait-free (or more generalized failure) models.
In particular, the immediate snapshot model
\cite{AttiyaRajsbaum02,BorowskyGafni:STOC93}
and the iterated immediate snapshot mode \cite{BorowskyGafni97}
are central to the study.
The (iterated) immediate snapshot protocol is modeled by
a topological operation on simplicial complexes, namely,
the (iterated) standard chromatic subdivision.
This topological interpretation has boosted theoretical investigations
on distributed systems using simplicial complexes.
Most notably, Herlihy and Shavit have
established the asynchronous computability theorem \cite{HerlihyShavit99},
which states that a distributed task is wait-free solvable
in the asynchronous read-write shared memory model
if and only if the task is expressed by a suitable pair of
an iterated standard chromatic subdivision and a simplicial map.

Although immediate snapshots are not natively supported in real distributed systems,
the protocol proposed by
Borowsky and Gafni \cite{BorowskyGafni:PODC93} provides a wait-free implementation
of it on asynchronous shared memory systems with atomic reads and writes.
Their protocol is a multi-round protocol, where
each individual process in a distributed system consisting of $n+1$ processes
decides its snapshot only if it witnesses, for $k$-th round,
$n+2-k$ different processes (including the process itself) that have written to shared memory.

In the present paper, we give yet another multi-round protocol
for the immediate snapshot, reformulating the one
by Borowsky and Gafni. Though both protocols work equally,
the reformulated version has
notable advantages over the original one.
\begin{enumerate}
  \item The reformulated multi-round protocol induces a neat correspondence
  of each individual round with a topological construction,
  namely a subdivision using Schlegel diagram.

  Benavides and Rajsbaum \cite{BenavidesRajsbaum16} showed that
  each round of the protocol by Borowsky and Gafni does not simply subdivides
  the input complex but it produces an intermediate protocol complex
  that is not a (pseudo)manifold.
  They make use of `collapsing' to describe how the protocol complex is
  transformed topologically at each protocol round, concluding that
  the standard chromatic subdivision is obtained as the final result.
  In contrast, the present paper shows that, up to reformulation,
  each protocol round exactly corresponds to a subdivision using Schlegel diagram.
  This series of corresponding subdivisions
  gives rise to the standard chromatic subdivision,
  due to a straightforward topological argument
  by Kozlov \cite{Kozlov12}.

  \item Due to the simpler topological structure,
  the reformulated version is amenable to mechanical optimization.

  In shared memory systems,
  shared memory access is a major bottleneck.
  Processes share a single shared memory module,
  which serializes simultaneous requests
  to handle them one at a time.
  The above mentioned multi-round protocols for the immediate snapshot
  involve read and write requests multiplied by the number of rounds to
  be performed (and also by the number of processes). The multiplied
  requests to the shared memory thus can cause
  performance degradation due to memory contention.

  For each concrete implementation of a protocol that makes use of the
  immediate snapshot,
  the reformulated version of the immediate snapshot protocol can be optimized
  to issue a lesser number of memory requests so that each process decides
  its output value at an earlier round, beyond which any execution path
  converges to a single unique output.
  This refinement is possible because the reformulated protocol
  iteratively subdivides the protocol complex at each round,
  narrowing down the candidates of output
  to those in a smaller region of finer subdivision.
  Thus the output can be decided as soon as the possible outputs
  have been narrowed down to a singleton set.
  This optimizes the protocol to perform a lesser number of shared memory access,
  where the optimized code can be mechanically derived
  for each concrete instance.

  With this optimization technique, the author believes that
  the immediate snapshot model, which
  has been studied mostly of theoretical concern,
  can also serve as a fundamental construct for wait-free distributed programming
  in a more practical context.

\end{enumerate}

\paragraph*{Related Work.}
In \cite{Kozlov12}, Kozlov has proven that
the standard chromatic subdivision is indeed a subdivision
by showing the standard chromatic subdivision is
obtained by the iterated subdivision using Schlegel diagram.
He also argued that the transient complexes
that appear in the intermediate steps of iterated subdivision
can be given computational interpretation,
but his interpretation combines atomic writes with
immediate scan operations.
In contrast, the present paper gives the computational interpretation
solely with atomic reads and writes, establishing
the exact correspondence of each round of the reformulated
immediate snapshot protocol
with subdivision using Schlegel diagram.

Benavides and Rajsbaum \cite{BenavidesRajsbaum16}
studied the topological structure induced by the multi-round
immediate snapshot protocol of Borowsky and Gafni.
They observed that the series of shared memory reads and writes
involved in each single round of the protocol
generates a protocol complex that augments
the standard chromatic subdivision with extra simplexes.
Due to those extras, the protocol complexes are neither a subdivision of
the input complex nor a (pseudo)manifold in general.
They analyzed the topological structure of the
protocol complexes in detail and presented a topological model,
in which subsequent rounds of the protocol collapse those extra simplexes
to end up with the standard chromatic subdivision.
This collapsing series of complexes, though topologically insightful,
does not explain well the correspondence to the underlying operational (execution) model.
Assuming the colored simplicial topological model, where
every simplex is a collection of vertexes of distinct colors (process ids),
there can be no operational counterpart to collapsing of a simplex:
Since a simplex is always collapsed to a degenerate simplex of strictly smaller
dimension (i.e., of fewer colors), it absurdly indicates that
such an operation would shrink a process group into a strictly smaller one
(even if no process in the group crashed).
In the present paper, we reformulate the multi-round protocol
by Borowsky and Gafni so that
each round is precisely a subdivision using Schlegel diagram,
providing a simpler topological model.
The neat correspondence between the topological model and the operational model
also allows us to optimize the protocol for reduced shared memory access.
(See Section~\ref{sec:optimiz}.)

Hoest and Shavit \cite{HoestShavit06} proposed
the nonuniform iterated immediate snapshot model,
a refinement of the immediate snapshot model, for the purpose of
a precise analysis of protocol complexity.
The nonuniform iterated immediate snapshot
topologically corresponds to the iterated nonuniform chromatic subdivision,
which generalizes the iterated standard chromatic subdivision so that
individual simplexes are allowed to be subdivided different numbers of times.
Their nonuniform model may also be applied to protocol optimization. That is,
a task solvable by a protocol in the uniform model would be substituted
by a protocol in the nonuniform model that solves the same task with
a coarser iterated nonuniform subdivision.
However, it seems nontrivial in general to find
an optimal nonuniform protocol.
The reformulated multi-round protocol in the present paper,
due to the smooth correspondence between the topological and operational models,
allows mechanical optimization on
any protocol given in the uniform model.

\paragraph*{Outline.}
The rest of the paper is organized as follows.
Section~\ref{sec:prelim} describes the shared memory
distributed computing model and the wait-free multi-round protocol by Borowsky
and Gafni and reviews the topological theory concerning wait-free solvability
of distributed tasks.
Section~\ref{sec:schlegel} proposes a reformulation
of the multi-round protocol by Borowsky and Gafni.
We prove the reformulated protocol also computes
the immediate snapshot, showing that each round of the reformulated protocol
precisely corresponds to a subdivision using Schlegel diagram.
In Section~\ref{sec:optimiz}, we further argue
that the reformulated protocol is amenable to optimization.
We show that, for each concrete protocol that solves a task
using the iterated immediate snapshot,
the protocol can be optimized to decide the final output
at an earlier round, as soon as the
collection of possible outputs to be reached
is narrowed down to a singleton set.
Finally, Section~\ref{sec:conclusion} concludes the paper.

\section{Distributed Computing and Topological Model}
\label{sec:prelim}

Throughout the paper we consider a distributed system of
$n+1$ faulty processes, which have process ids numbered
$0$ through $n$. We assume the asynchronous read-write shared memory model,
where the distributed system has a shared memory
consisting of single-writer, multi-reader atomic registers
and processes can communicate solely by atomic reads and writes
on these registers.
We further assume that each process $i$ receives
its initial private input value through the variable $v_i$,
which is local to the process.

The rest of this section is devoted to give
an overview of the immediate snapshot model
and the basics of combinatorial topology
related to the topological theory of wait-free solvability
of distributed tasks.
For a more complete exposition on the subject,
see \cite{Herlihy:DCTopology,Kozlov:CombiAlgTopology}.

\subsection{Immediate snapshot in the read-write model}
\label{subsec:ISprotocol}

Borowsky and Gafni proposed the wait-free implementation
of the immediate snapshot in the read-write shared memory model \cite{BorowskyGafni:PODC93}.
Algorithm~\ref{algo:IS} gives
the Borowsky-Gafni protocol in a recursive style implementation
by Gafni and Rajsbaum \cite{GafniRajsbaum10}.
The protocol is a multi-round protocol,
where the series of recursive calls $\ISnap(n)$,
$\ISnap(n-1)$, \ldots, $\ISnap(0)$ correspond to $n+1$ multiple rounds
and the processes communicate through a series of $n+1$
shared memory arrays $\mathit{mem}_d$ ($n\geq d\geq 0$),
each of which consists of
$n+1$ registers indexed $0$ through $n$.
Each process $i$,
per each recursive call $\ISnap(d)$ for the $(n+1-d)$-th round,
invokes the \keyword{write\&scan} operation $\WScan$
on the shared memory array $\mathit{mem}_d$.
When $\WScan$ is invoked by process $i$,
it first writes the private value $v_i$ of the process $i$
to the register $\mathit{mem}_d[i]$,
then scans the view, i.e.,
the set of values that have been written in the array $\mathit{mem}_d$,
and returns the view paired with the process id.
The view is gathered by $\textproc{collect}$,
which reads the registers of the array in an unspecified order.
If process $i$ witnesses $d+1$ values in the view, it returns
the pair of process id and the view as the result of immediate snapshot,
terminating recursion;
Otherwise, it continues recursive call $\ISnap(d-1)$ for another round.

We notice that the view returned by $\WScan$ per each round is simply discarded,
when not sufficiently many values are witnessed.

\subsection{Simplicial complexes and subdivisions}
\label{subsec:simpcpl}

Let $V$ be a set of vertexes.
A \keyword{simplex} $\sigma$ is a finite subset of $V$.
The dimension of $\sigma$, denoted by $\dim(\sigma)$, is given by $\abs{\sigma} -1$.
A simplex $\sigma$ of dimension $d$ is called a
\keyword{$d$-simplex}. %
In particular, the empty simplex $\emptyset$ is a $(-1)$-simplex.
A simplex $\sigma$ is called a \keyword{face}
of $\tau$ if $\sigma\subseteq \tau$
and is particularly called a \keyword{proper} face
if the inclusion is strict.

A \keyword{simplicial complex} (or a \keyword{complex} for short) $\complexf{C}$
is a set of %
simplexes such that %
$\sigma\in \complexf{C}$ and $\tau\subseteq \sigma$
implies $\tau\in \complexf{C}$.
The dimension of $\complexf{C}$, written $\dim(\complexf{C})$,
is the maximum dimension of simplexes contained in $\complexf{C}$.
A complex of dimension $d$ is also called a \keyword{$d$-complex}.
We write $V(\complexf{C})$ for the set of vertexes contained in $\complexf{C}$.
A complex $\complexf{D}$ is called a \keyword{subcomplex} of $\complexf{C}$, if
$\complexf{D}\subseteq \complexf{C}$.
The \keyword{$k$-skeleton} of a complex $\complexf{C}$,
written $\skel{k}\complexf{C}$, is the maximum
subcomplex of $\complexf{C}$ of dimension $k$ or less,
namely, $\skel{k}\complexf{C} =
\{\sigma \in\complexf{C}\mid \dim(\sigma)\leq k \}$.

\begin{algorithm}[t]
  \caption{Multi-round write\&scan code for the process $i$}
  \label{algo:IS}
  \begin{algorithmic}
    \Procedure{WScan}{$d$}
    \State $mem_d[i] \gets v_i$
    \State $view \gets \Call{collect}{mem_d}$
    \State \Return $(i,view)$
    \EndProcedure
  \end{algorithmic}
  \begin{algorithmic}
    \Procedure{IS}{$d$}
    \State $(i,view) \gets \Call{WScan}{d}$
    \If{$\abs{view}= d+1$}
    \Return $(i,view)$
    \Else~\Call{IS}{$d-1$}
    \EndIf
    \EndProcedure
  \end{algorithmic}
\end{algorithm}

A \keyword{facet} $\sigma$ of $\complexf{C}$ is a maximal simplex, i.e.,
$\sigma$ is not a proper face of any $\tau\in\complexf{C}$.
We write $\cls{\sigma}$ for a complex whose sole facet is the simplex $\sigma$,
i.e.,
$\cls{\sigma} = \{ \tau \mid \tau \subseteq \sigma \}$.
A complex is called \keyword{pure},
if all its facets have the same dimension.

Whenever $\sigma\cap\tau = \emptyset$,
we write $\sigma\join\tau$ for
the \keyword{join} of the simplexes $\sigma$ and $\tau$,
namely the union $\sigma\cup\tau$.
Likewise, whenever
$V(\cplC)\cap V(\cplD)=\emptyset$,
we define the \keyword{join} of $\complexf{C}$ and $\complexf{D}$
by $\cplC\join\cplD = \{\sigma\join\tau \mid \sigma\in\cplC, \tau\in\cplD\}$.
The \keyword{star} of a vertex $v\in\cplC$ is
defined by $\Star(v,\complexf{C}) = \{\tau\in\complexf{C}
\mid \{v\}\cup\tau\in \complexf{C}\}$,
which is the maximum subcomplex of $\cplC$ whose every facet contains $v$.

Throughout the paper, complexes are assumed to be pure.
Also, we solely consider the so-called \keyword{chromatic} simplexes
and complexes. Suppose we have a \keyword{coloring function}
$\coloring: V \to \{0,\ldots, n\}$. A simplex $\sigma$ is chromatic
if $\coloring(v)=\coloring(v')$ implies $v=v'$ for every $v,v'\in\sigma$;
A complex $\cplC$ is chromatic if every simplex $\sigma\in\cplC$ is chromatic.

A \keyword{simplicial map} is a
total vertex map $\mu: V(\complexf{C}) \to V(\complexf{D})$
such that $\mu(\sigma)\in \complexf{D}$ for every $\sigma\in\complexf{C}$.
Every simplicial map must be \keyword{color-preserving},
i.e., $\coloring(v)=\coloring(\mu(v))$ for every
$v\in V(\complexf{C})$.

A \keyword{subdivision} of a complex $\complexf{C}$,
written $\Div{Div}\,\complexf{C}$,
is a finer complex obtained by
dividing each simplex of the complex into smaller pieces
of (chromatic) simplexes.
(In this paper, though we occasionally refer to geometric presentation
of subdivisions,
we are solely concerned with combinatorial definition
of subdivision in formality. For example,
$\cls{\sigma}$ should be understood as the trivial subdivision,
which does not geometrically refine $\sigma$ at all.)
For a simplex $\tau\in\Div{Div}\,\complexf{C}$,
we write $\Carrier(\tau,\complexf{C})$
for the \keyword{carrier} of $\tau$, namely,
the smallest simplex $\sigma\in\complexf{C}$ such that
$\tau\in \Div{Div}\,\sigma$.

A subdivision induces a \keyword{parent map} $\pmap{}$, which
carries each vertex of the subdivision to the unique vertex
$\pmap{}(v)$ of matching color in its carrier.
That is, for $v\in V(\Div{Div}\,\cplC)$, $\pmap{}(v)$ is the unique vertex
of $\cplC$ such that
$\pmap{}(v)\in\Carrier(\{ v\},\cplC)$ and
$\coloring(\pmap{}(v))=\coloring(v)$.
The parent map $\pmap{}: V(\Div{Div}\,\cplC) \to V(\cplC)$ is
a color-preserving simplicial map.

\subsection{Distributed task and asynchronous computability}
\label{subsec:topmethods}

A (colored) \keyword{task} for a distributed system with $n+1$ faulty processes
is defined by a triple $(\cplI, \cplO, \Phi)$,
where $\cplI$ is an input complex (of dimension $n$),
$\cplO$ is an output complex (of dimension $n$),
and $\Phi: \cplI \to 2^{\cplO}$ is a \keyword{carrier map},
a monotonic function that maps every input simplex $\sigma\in\cplI$
to an output subcomplex $\Phi(\sigma)\subseteq\cplO$ such that
$\coloring(\sigma)=\bigcup\{\coloring(\tau)\mid\tau\in\Phi(\sigma)\}$.

\begin{figure}[t]
  \qquad\includegraphics[scale=0.3]{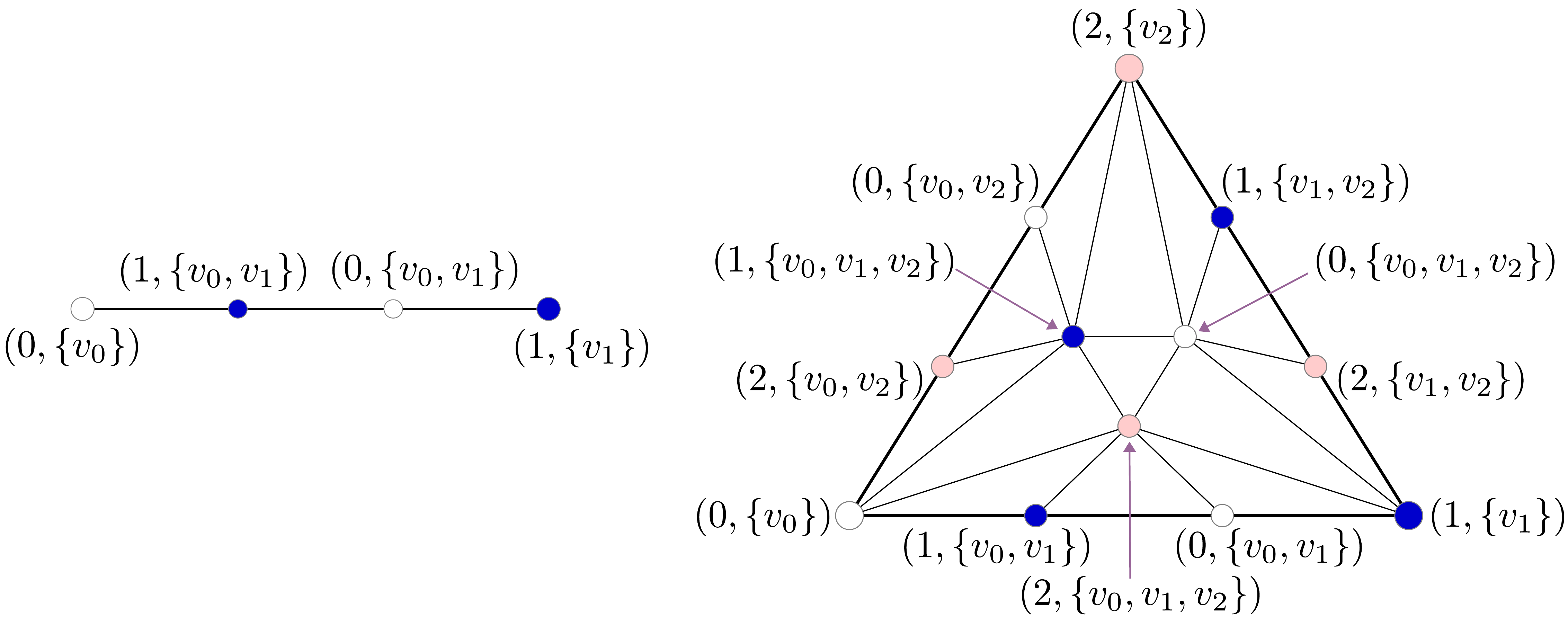}
  \caption{The standard chromatic subdivisions on 1-simplex and 2-simplex}
  \label{fig:Chsubdiv}
\end{figure}

The asynchronous computability theorem \cite{HerlihyShavit99}
gives a topological characterization for the class of
tasks that are wait-free solvable in the shared memory read-write model.
The primary topological tool employed in the
theorem is subdivision on simplicial complexes, especially
the \keyword{standard chromatic subdivision}.
See Figure~\ref{fig:Chsubdiv} for geometric presentation
of the standard chromatic subdivision on 1-simplex $\{v_0,v_1\}$
and 2-simplex $\{v_0,v_1,v_2\}$:
The standard chromatic subdivision refines each $d$-simplex $\sigma$
($d>0$) by introducing new $d+1$ vertexes (of different colors)
at antiprismatic positions displaced from
the barycenter of the simplex.
Combinatorially, as indicated in Figure~\ref{fig:Chsubdiv},
each vertex of color $i$ contained
in the subdivision is designated by a pair $(i,\sigma)$, where
$\sigma$ is a nonempty subset of the original simplex and
$i\in\coloring(\sigma)$.
(Herein and after we will depict the vertexes for different processes
with colors white, dark blue, and light red to indicate
processes of id 0, 1, and 2, respectively.)
Furthermore, a set $\{(i_0,\tau_0),\ldots,(i_d,\tau_d)\}$ of vertexes
of distinct colors forms a $d$-simplex of the subdivision if and only if
the subsets of the original simplex are linearly ordered
(as $\tau_0\subseteq\cdots\subseteq\tau_d$, up to appropriate permutation).

In what follows, let us write $\Chromatic\,\cplC$ for the standard
chromatic subdivision applied to every simplex in $\cplC$
and also write $\Chromatic^m\cplC$ ($m\geq 0$) for the $m$-iterated application
of $\Chromatic$ on $\cplC$.

The multi-round protocol for the immediate snapshot (Algorithm~\ref{algo:IS})
in effect implements the standard chromatic subdivision \cite{BenavidesRajsbaum16}:
Every non-faulty process $i$ executing the protocol returns
a vertex $(i,\tau)$ of the subdivision.
\begin{theorem}[Theorem~5.29 and Corollary~5.31 of \cite{HerlihyShavit99}]
  A colored task $(\cplI,\cplO,\Phi)$ has a wait-free protocol
  in the asynchronous shared memory read-write model
  if and only if there exists a subdivision $\Div{Div}\,\cplI$ and
  a \keyword{decision map} $\delta$, which is
  a color-preserving simplicial map $\delta:V(\Div{Div}\,\cplI) \to V(\cplO)$
  such that $\delta(\sigma)\in\Phi(\Carrier(\sigma,\cplI))$ for
  every $\sigma\in\Div{Div}\,\cplI$.

  In particular, the subdivision $\Div{Div}$ can be taken
  $\Chromatic^K$ for some $K\geq 0$.
\end{theorem}

\section{Immediate Snapshot as Iterated Subdivision Using Schlegel Diagram}
\label{sec:schlegel}

\begin{figure}[t]
  \quad\qquad\includegraphics[scale=0.3]{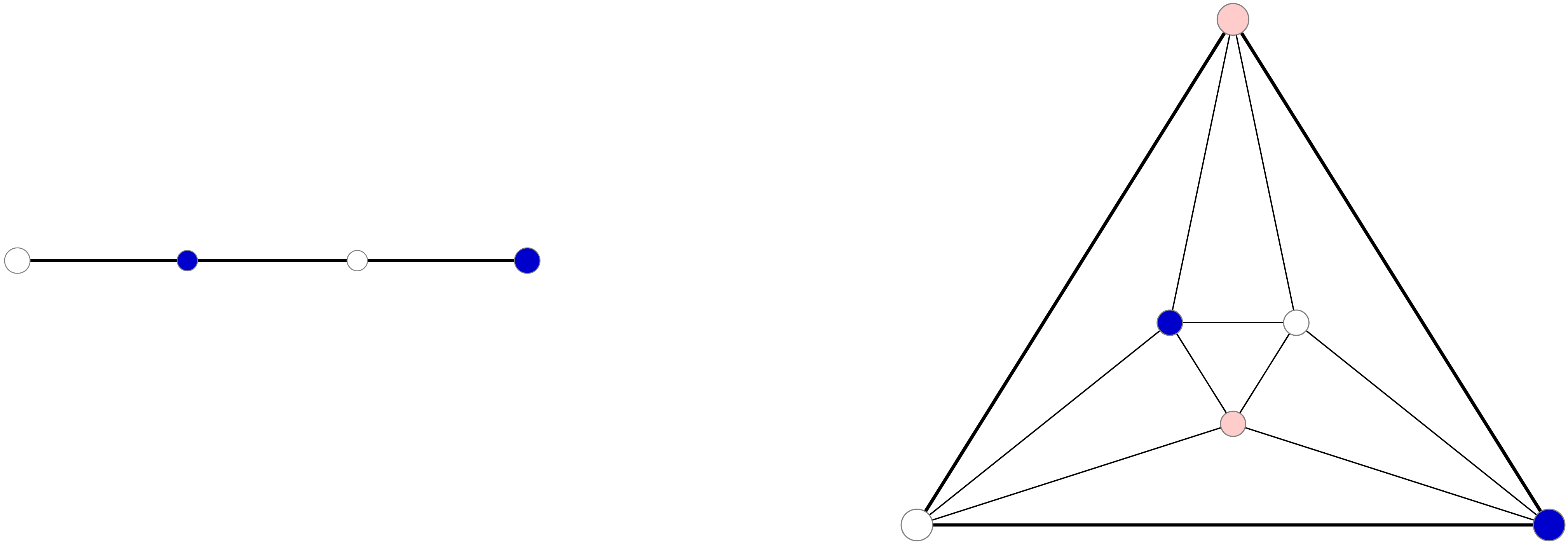}
  \caption{Schlegel diagrams for $1$-simplex and $2$-simplex}
  \label{fig:sch}
\end{figure}

This section presents a multi-round protocol for the immediate snapshot,
by reformulating the protocol by
Borowsky and Gafni \cite{BorowskyGafni:PODC93}.
We show that each round of the protocol
corresponds exactly to a subdivision using Schlegel diagram.

\subsection{Schlegel diagram and subdivision}
\label{subsec:SchSubdiv}

A \keyword{Schlegel diagram} is a projection of
a polytope onto one of its facets \cite{Ziegler95}.
In the present paper,
we are solely concerned with Schlegel diagrams on cross-polytopes.
The Schlegel diagram on $(d+1)$-dimensional cross-polytope,
which consists of $2(d+1)$ vertexes, gives a subdivision of $d$-simplex.
Figure~\ref{fig:sch} shows subdivisions of $1$-simplex and $2$-simplex
by Schlegel diagrams, where
the former is derived from the quadrilateral %
and the latter from the octahedron.
(Note that the standard chromatic subdivision is a refinement of
Schlegel diagram in general, but they coincide for $1$-simplexes.)

The complex of Schlegel diagram that subdivides
a $d$-simplex $\sigma=\{v_0, \ldots, v_d \}$,
denoted by $\Sch{}{d}\,\sigma$, is formally defined
as follows:
\begin{align*}
\Sch{}{d}\,\sigma = &
\bigcup\bigl\{ \cls{\{v_i\mid i\in \sigma\setminus I\}\join\{ (i,\sigma) \mid i\in I\}} \mid
\emptyset \subsetneq I\subseteq \{0,\ldots,d\}
\bigr\},
\end{align*}
where each $(i,\sigma)$ is a new vertex introduced for subdivision,
with coloring $\coloring((i,\sigma))=i$.
Geometrically, the Schlegel diagram subdivides a $d$-simplex
$\sigma$ into smaller facets, namely,
the central facet $\{(0,\sigma),\ldots,(d,\sigma)\}$, which is solely comprised
of the new vertexes, and other facets, each of which
shares a lower dimensional simplex with the central facet.
(See Figure~\ref{fig:ChFromSchlegel} for the subdivision of 2-simplex.)
Notice that $\sigma$ is no more a face of Schlegel subdivision
$\Sch{}{d}\,\sigma$, as it does not share any simplex with the central facet.
When $\dim(\sigma)\neq d$, we define $\Sch{}{d}\,\sigma$ by a trivial
subdivision, i.e., $\Sch{}{d}\,\sigma = \cls{\sigma}$.

\begin{figure}[t]
  \includegraphics[scale=0.43]{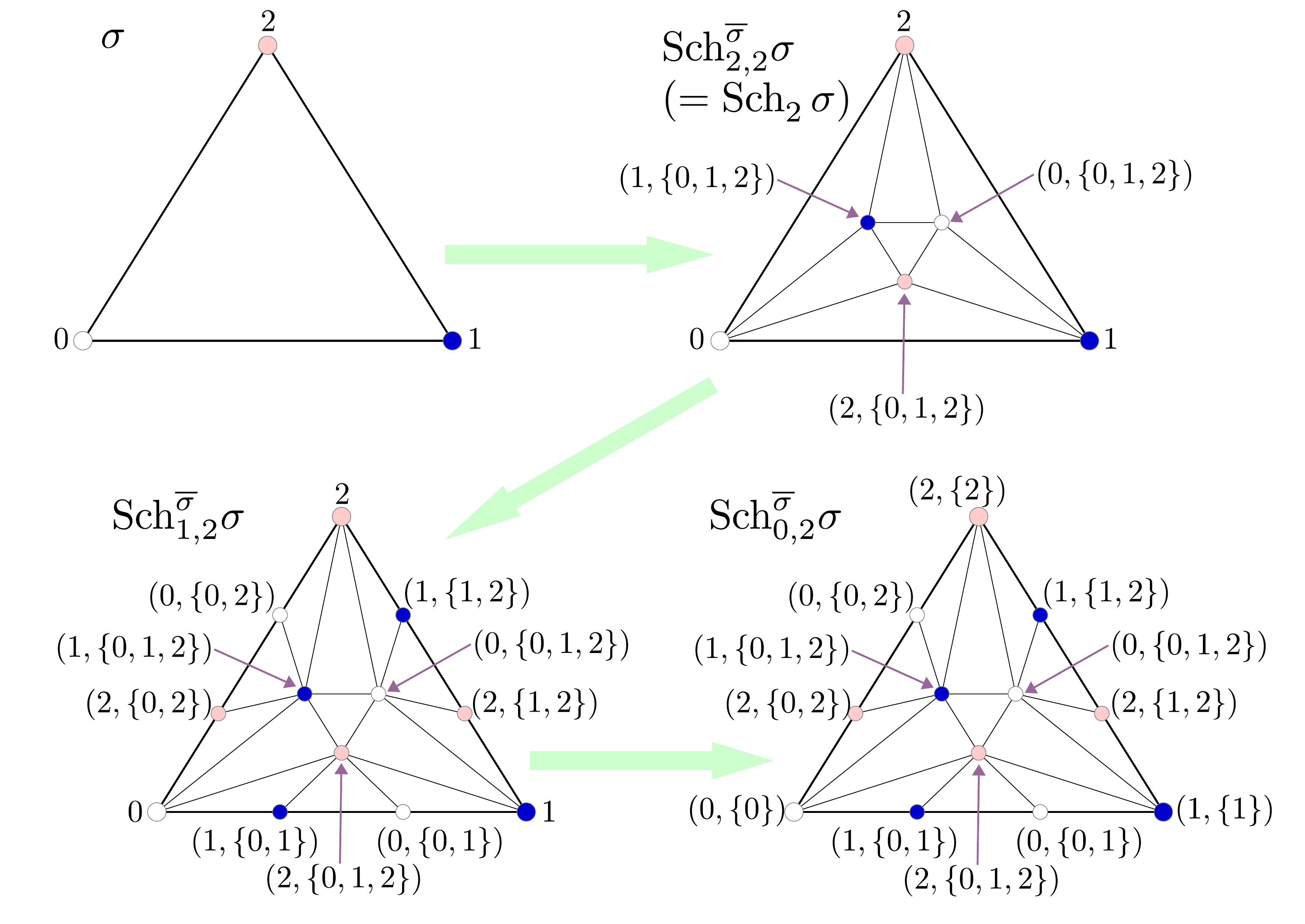}
  \caption{Standard chromatic subdivision on $\sigma=\{0,1,2\}$ using Schlegel diagrams}
  \label{fig:ChFromSchlegel}
\end{figure}

Kozlov \cite{Kozlov12} has shown that
the standard chromatic subdivision can be obtained by
a series of Schlegel diagram that subdivides simplexes in the order
of decreasing dimension.
To put it formal, let $\cplD$ be any subdivision of a $d$-complex $\cplC$
such that $\cplC\cap \cplD= \skel{k}\cplC$, meaning that
$\cplD$ subdivides simplexes of $\cplC$ up to dimension
$k+1$ and higher but no simplexes of lower dimension.
Let us write $\Sch{\cplC}{k}\cplD$ for the subdivision of $\cplD$ applied to
every $k$-simplex of $\cplC\cap \cplD$ by Schlegel diagram, namely,
\[
\Sch{\cplC}{k}\cplD = \{
\tau\join\sigma' \mid
\tau\join\sigma\in\cplD \text{ and }
\sigma'\in\Sch{}{k}\,\sigma
\text{ for some } \sigma\in\cplC\cap\cplD\}.
\]
Let us also write $\Sch{\cplC}{h,j}$
for the composition $\Sch{\cplC}{h}\circ\Sch{\cplC}{h+1}\circ\cdots\circ\Sch{\cplC}{j}$
($0\leq h, j\leq d$) of subdivisions on simplexes in the order of decreasing dimension.
(When $h>j$, $\Sch{\cplC}{h,j}$ denotes the trivial subdivision.)

\begin{theorem}[\cite{Kozlov12}] \label{th:ChByIteratedSchlegel}
  For any pure complex $\cplC$ of dimension $d$,
  $\Chromatic\,\cplC = \Sch{\cplC}{0,d}\,\cplC$.
\end{theorem}

Figure~\ref{fig:ChFromSchlegel} shows how the standard chromatic subdivision
on a 2-simplex $\sigma=\{0,1,2\}$ with $\coloring(i)=i$ for every $i\in\{0,1,2\}$
is obtained by the series of subdivisions using Schlegel diagram.
For example, the 2-simplex $\{0,1\}\join \{(2,\{0,1,2\})\}$ in
$\Sch{\cls{\sigma}}{2,2}$ is further refined in the next step of subdivision,
say, the 1-simplex $\{0,1\}$ is subdivided into three parts
$\{0,(1,\{0,1\})\}$,
$\{(0,\{0,1\}),(1,\{0,1\})\}$,
$\{(0,\{0,1\}),1\}$, which are each \emph{joined} with $\{(2,\{0,1,2\})\}$.

\subsection{The immediate snapshot protocol with oblivious scan}
\label{subsec:oblivScan}

\begin{algorithm}[t]
  \caption{Multi-round write\&oblivious scan code for the process $i$}
  \label{algo:IobS}
  \begin{algorithmic}
    \Procedure{WOScan}{$d$}
    \State $mem_d[i] \gets v_i$
    \State $view \gets \Call{collect}{mem_d}$
    \If{$\abs{view}=d+1$} \Return $(i,view)$
    \Else~\Return $v_i$
    \EndIf
    \EndProcedure
  \end{algorithmic}
  \begin{algorithmic}
    \Procedure{IS'}{$d$}
    \State $u \gets \Call{WOScan}{d}$
    \If{$u\neq v_i$}
    \Return $u$
    \Else~\Call{IS'}{$d-1$}
    \EndIf
    \EndProcedure
  \end{algorithmic}
\end{algorithm}

Algorithm~\ref{algo:IobS} gives a multi-round immediate snapshot protocol,
which reformulates Algorithm~\ref{algo:IS} in Section~\ref{subsec:ISprotocol}.
It is easy to see that they are indeed equivalent protocols in different presentations.
Remember that in Algorithm~\ref{algo:IS} the view information
collected at each particular round is discarded unless the view
witnesses the expected number of writes. Algorithm~\ref{algo:IobS} just makes
this explicit by employing the \keyword{write\&oblivious scan} operation $\WOScan$
on shared memory array, in place of write\&scan operation.
When process $i$ calls $\WOScan(d)$ and the view does not witness
$d$ writes,  $\WOScan(d)$ returns $v_i$, discarding the view collected at
the scan phase.

With this reformulation, we can prove that
the protocol computes the standard chromatic subdivision,
with an exact correspondence of
a write\&oblivious scan operation at a particular round
with Schlegel diagram.

\begin{lemma} \label{lem:WOScanSchlegel}
  Suppose $\WOScan(d)$ has ever been called by $d+1$ distinct processes.
  If $\sigma=\{v_0, \ldots, v_d \}$ is the collection of private values
  that have been assigned to the $d+1$ processes and
  $\tau$ is the set of results returned by non-faulty processes,
  then $\tau \in \Sch{}{d}\,\sigma$.
\end{lemma}
\begin{proof}
  When $d+1$ processes invoked $\WOScan(d)$ and
  none of them were faulty,
  at least one process witnesses the writes by all the $d+1$ processes
  in its view. This means $\tau\setminus \sigma\neq\emptyset$
  and thus $\tau\in \Sch{}{d}\,\sigma$. When some of the processes
  were faulty, $\dim(\tau)<d$ and $\tau\in \Sch{}{d}\,\sigma$.
\end{proof}

\begin{lemma} \label{lem:roundSchlegel}
  Consider an execution of the protocol $\ISnapO(n)$ by $n+1$ processes,
  in which each process $i$ has started with its own initial private value $v_i$
  and has either successfully returned a result or crashed.
  Let $\sigma_d$ ($n\geq d\geq 0$) denote the set of results that have been
  successfully returned by a call $\WOScan(d)$ by some process in the execution.
  Also, let $\tau_d$ ($n\geq d\geq 0$) denote the set of values that have
  been returned by a call $\ISnapO(k)$ for some $k$ greater than $d-1$.
  We define $\sigma_{n+1}= \{v_0,\ldots,v_n\}$ and $\tau_{n+1}=\emptyset$.

  Then the following properties hold
  for every $d$ $(n+1\geq d\geq 0)$.
  \begin{enumerate}[(i)]
    \item \label{lem:rSch:dim}
    $\dim(\sigma_{d}\cap\sigma_{n+1})<d$;
    \item \label{lem:rSch:decids}
    $\tau_d = \tau_{d+1}\join (\sigma_d\setminus\sigma_{n+1})$
    and $\tau_d\cap \sigma_{n+1}= \emptyset$;
    \item \label{lem:rSch:simplex}
    $\tau_{d+1}\join \sigma_d \in \Sch{\cls{\sigma_{n+1}}}{d,n}\sigma_{n+1}$.
  \end{enumerate}
\end{lemma}
\begin{proof}
  The property \eqref{lem:rSch:dim} follows from lemma~\ref{lem:WOScanSchlegel}
  by induction on $d$.
  The property \eqref{lem:rSch:decids} immediately follows from the definition
  by an inductive argument.

  Let us show \eqref{lem:rSch:simplex} by induction on $d$.
  For the base case $d=n$, since
  $\sigma_n \in \Sch{}{n}\sigma_{n+1}$ by lemma~\ref{lem:WOScanSchlegel},
  we have $\tau_{n+1}\join \sigma_n=\sigma_n\in
  \Sch{\cls{\sigma_{n+1}}}{n,n}\sigma_{n+1}$.
  For the inductive step, assume
  $\tau_{d+2}\join \sigma_{d+1}\in\Sch{\cls{\sigma_{n+1}}}{d+1,n}\sigma_{n+1}$.
  By lemma~\ref{lem:WOScanSchlegel} we have
  $\sigma_d \in \Sch{}{d}(\sigma_{d+1}\cap\sigma_{n+1})$.
  Since $\sigma_{d+1}\cap\sigma_{n+1}\in\cls{\sigma_{n+1}}$,
  we have
  $\tau_{d+1}\join (\sigma_{d+1}\cap\sigma_{n+1})  =
  \tau_{d+2}\join(\sigma_{d+1}\setminus\sigma_{n+1})\join
  (\sigma_{d+1}\cap\sigma_{n+1})
  = \tau_{d+2}\join\sigma_{d+1}
  \in
  \Sch{\cls{\sigma_{n+1}}}{d+1,n}\,\sigma_{n+1}$ by property~\eqref{lem:rSch:decids}
  and the induction hypothesis.
  Hence $\tau_{d+1}\join\sigma_d \in
  \Sch{\cls{\sigma_{n+1}}}{d}(\Sch{\cls{\sigma_{n+1}}}{d+1,n}\,\sigma_{n+1})
  =\Sch{\cls{\sigma_{n+1}}}{d,n}\,\sigma_{n+1}$.
\end{proof}

\begin{theorem} \label{th:ISprotocol}
  Suppose $n+1$ processes executed the protocol $\ISnapO(n)$
  with the set $\sigma=\{v_0,\ldots,v_n\}$
  of initial private inputs.
  If $\tau$ is the set of results returned by non-faulty processes,
  $\tau\in\Chromatic\,\sigma$.
\end{theorem}

\begin{proof}
  Let $\sigma_d$'s and $\tau_d$'s denote the sets as defined in
  lemma~\ref{lem:roundSchlegel}. Then, $\sigma=\sigma_{n+1}$ and
  $\tau=\tau_1\join \sigma_0$.
  By lemma~\ref{lem:roundSchlegel}\eqref{lem:rSch:simplex} and
  theorem~\ref{th:ChByIteratedSchlegel},
  we have
  $\tau=\tau_1\join \sigma_0 \in \Sch{\cls{\sigma}}{0,n}\sigma
  =\Chromatic\,\sigma$.
\end{proof}

\section{The Generic Protocol for Solving Tasks and Its Optimization}
\label{sec:optimiz}

This section gives a generic protocol for solving a task
on read-write shared memory distributed system,
using the immediate snapshot protocol presented in the previous section,
and discusses how, for each concrete instance, the generic protocol
can be optimized to reduce shared memory access.

\subsection{A generic protocol via iterated immediate snapshot}
\label{subsec:genericProtocol}

By the asynchronous computability theorem,
for any wait-free solvable task, we have a protocol
$(\cplI,\cplO,\delta\circ\Chromatic^K)$ that implements the task,
where the carrier map is given by a pair
of the full-information protocol of
the $K$-iterated standard chromatic subdivision $\Chromatic^K$
(by means of the iterated use of the multi-round immediate snapshot protocol)
and a decision map $\delta: V(\Chromatic^K) \to V(\cplO)$.
Without loss of generality, we may assume $K\geq 1$.

\begin{algorithm}
  \caption{The generic code for the process $i$, using iterated immediate snapshot}
  \label{algo:IIS}
  \begin{algorithmic}
    \Procedure{WOScan}{$k,d$}
    \State $mem_{k,d}[i] \gets v_i$
    \State $view \gets \Call{collect}{mem_{k,d}}$
    \If{$\abs{view}=d+1$} \Return $(i,view)$
    \Else~\Return $v_i$
    \EndIf
    \EndProcedure
  \end{algorithmic}
  \begin{algorithmic}
    \Procedure{IIS}{$k,d$}
    \State $u \gets \Call{WOScan}{k,d}$
    \If{$u\neq v_i$}
      \If{$k=K$} \Return $\delta(u)$
      \Else~$v_i \gets u$; $\Call{IIS}{k+1,n}$
      \EndIf
    \Else~\Call{IIS}{$k,d-1$}
    \EndIf
    \EndProcedure
  \end{algorithmic}
\end{algorithm}

Algorithm~\ref{algo:IIS} gives the code that implements the protocol
$(\cplI,\cplO,\delta\circ\Chromatic^K)$ in the generic form.
Each process $i$ initiates the protocol execution
by invoking $\IISnap(1,n)$, where
its initial private input is passed through the variable $v_i$.
Throughout the entire protocol execution, each process goes through
a series of shared memory arrays
$mem_{k,d}$ ($1\leq k\leq K$, $n\geq d\geq 0$).
For each recursive call $\IISnap(k,d)$,
process $i$ computes the $(n-d+1)$-th round of the $k$-th iteration
of the multi-round immediate snapshot protocol,
by invoking the write\&oblivious scan $\WOScan(k,d)$
on the array $mem_{k,d}$.
Each round corresponds to a single step of
subdivision on $d$-simplexes using Schlegel diagram, for the $k$-th
iteration of standard chromatic subdivision.
When the multi-round execution by process $i$ finishes the last iteration
of chromatic subdivision (i.e., $k=K$),
the protocol returns $\delta(u)$ as the output, where $u$ is a vertex
of the $K$-iterated standard chromatic subdivision.

The following is a corollary to Theorem~\ref{th:ISprotocol}.
\begin{theorem} \label{th:IIS}
  Let $(\cplI,\cplO,\delta\circ\Chromatic^K)$ be a protocol
  for $n+1$ processes that solves a task.
  Suppose each process $i$ starts with a private input value $v_i$
  such that   $\sigma=\{v_0,\ldots,v_n\}\in\cplI$
  and executes the protocol by invoking $\IISnap(1,n)$.
  If $\tau$ is the set of results returned by non-faulty processes,
  $\tau\in \delta(\Chromatic^K\sigma)$.
\end{theorem}

\subsection{Protocol optimization for reduced memory access}
\label{subsec:optProto}

Algorithm~\ref{algo:IIS}
gives a generic protocol, but
a concrete instance of it %
often contains redundant memory access.
Below we argue that each instance of the generic protocol can be
mechanically optimized to skip the redundant access,
by applying the technique of program specialization.
(Program specialization is a source-level
program optimization technique, also known as
\keyword{partial evaluation} \cite{JonesGomardSestoft:PEbook}.
See Appendix~\ref{sec:PE} for a brief overview.)

To see how the protocol is optimized, let us consider a
particular instance of the generic protocol that solves
a renaming task \cite{BorowskyGafni:PODC93} for 3 processes, \unskip\footnote{%
  This particular instance of renaming task is taken from
  \cite[Chapter~12]{Herlihy:DCTopology}.}
where the 3 processes, starting with initial assignment
$v_0=0$, $v_1=1$, $v_2=2$, respectively,
decide on different names taken from $\{ 0,1,2,3,4 \}$.
The protocol is given by
$(\cplI,\cplO,\delta \circ \Chromatic^2)$,
where
$\cplI = \cls{\{ 0,1,2\}}$ is the input complex, %
$\cplO = \{ \tau \mid \tau\subseteq\{ 0,1,2,3,4 \},
\dim(\tau)\leq 2 \}$ is the output complex of
differently renamed vertexes,
and $\delta: V(\Chromatic^2\,\cplI) \to V(\cplO)$ is the decision map
defined with the corresponding parent map
$\pmap{2,1}: V(\Chromatic^2\,\cplI) \to V(\Chromatic\,\cplI)$ as given below:
\[
  \delta(w) =
  \begin{cases}
    \quad 4
    & \text{if $\pmap{2,1}(w)=(0,\{0,1,2\})$,}
    \\
    \quad 3
    & \text{\parbox[t]{.78\textwidth}{%
    if $\Carrier(w,\Chromatic\,\cplI)\not\supseteq\{(1,\{0,1,2\}),(2,\{0,1,2\})\}$\\
    and $\pmap{2,1}(w)=(i,\{0,1,2\})$ for some $i\in \{1,2\}$,
    }}
    \\
    \quad 2
    & \text{\parbox[t]{.78\textwidth}{%
    if either $\coloring(w)=1$ and $\Carrier(w,\Chromatic\,\cplI)\supseteq\{(1,\{0,1,2\}),(2,\{0,1,2\})\}$\\
    or $\pmap{2,1}(w)=(i,\{i,j\})$ for some $i,j\in \{0,1,2\}$ such that $i<j$,
    }}
    \\
    \quad 1
    & \text{\parbox[t]{.78\textwidth}{%
    if either $\coloring(w)=2$ and $\Carrier(w,\Chromatic\,\cplI)\supseteq\{(1,\{0,1,2\}),(2,\{0,1,2\})\}$\\
    or $\pmap{2,1}(w)=(i,\{i,j\})$ for some $i,j\in \{0,1,2\}$ such that $i>j$,
    }}
    \\
    \quad 0
    & \text{\parbox[t]{.78\textwidth}{%
    if $\pmap{2,1}(w)=(i,\{i\})$ for some $i\in \{0,1,2\}$.
    }}
  \end{cases}
\]
In the definition above, it is assumed that
each vertex is appropriately colored according to the context.
In particular, as every simplex $\tau\in\cplO$ of renamed processes
is colored, $V(\cplO)$ comprises 15 vertexes, namely, 3 differently colored vertexes
per each name taken from $\{ 0,1,2,3,4 \}$. Accordingly,
as $\delta$ is a color-preserving simplicial map,
the renamed output $\delta(w)$ for each $w\in V(\cplO)$
is tacitly given the matching color,
i.e., $\coloring(w)$.

\begin{figure}[t]
  \begin{center}
  \begin{minipage}[t]{10cm}
  \includegraphics[scale=0.42]{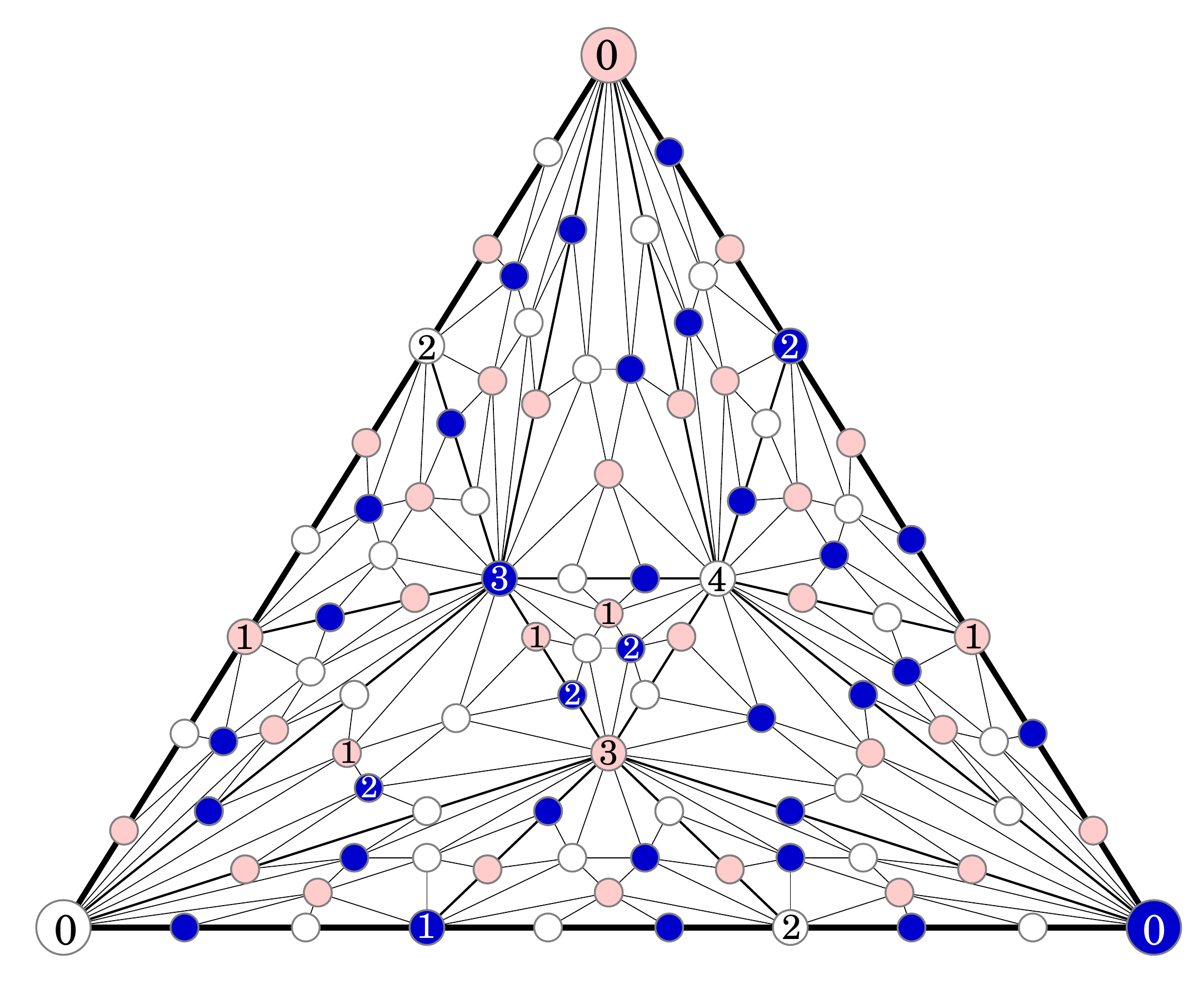}
  \end{minipage}
  \caption{Decision map $\delta$ for renaming}
  \label{fig:renamDecision}
  \end{center}
\end{figure}

Figure~\ref{fig:renamDecision} shows how
the decision map $\delta$ assigns an output to each
vertex of $\Chromatic^2\,\cplI$.
For those vertexes whose outputs are left unspecified in the figure,
it should be understood that
such a vertex $w$
receives the output $\delta(\pmap{2,1}(w))$,
namely the same output as the parent vertex $\pmap{2,1}(w)\in\Chromatic\,\cplI$ does.
For instance, the white vertex (of process 0)
of the very central simplex of the subdivision is assigned
the output $4$ by $\delta$, as its parent is $(0,\{0,1,2\})$,
the white vertex introduced by the first subdivision using Schlegel diagram.

Let us consider an execution of the generic protocol (Algorithm~\ref{algo:IIS})
for the renaming task. In the execution, each process~$i$
executes a chain of recursive calls of $\IISnap$,
where each recursive call updates $v_i$ to a vertex of a finer subdivision.
For instance, consider
the following particular recursive call chain for process~$2$:
\begin{align*}
2\in \cplI \xrightarrow{~\IISnap(1,2)~} &
2\in \Sch{\cplI}{2,2}\cplI
\xrightarrow{~\IISnap(1,1)~}
(2,\{1,2\})\in \Sch{\cplI}{1,2}\cplI \\
\xrightarrow{~\IISnap(2,2)~} &
(2,\{1,2\})\in \Sch{\Chromatic\,\cplI}{2,2}(\Chromatic\,\cplI) \\
\xrightarrow{~\IISnap(2,1)~} &
(2,\{(0,\{0,1,2\}),(2,\{1,2\})\})\in\Sch{\Chromatic\,\cplI}{1,2}(\Chromatic\,\cplI),
\end{align*}
where each transition
$u\in \Div{Div}\,\cplI \xrightarrow{~\IISnap(k,d)~} u'\in \Div{Div}'\cplI$
indicates that a recursive call $\IISnap(k,d)$
updates $v_2$ from $u$ to $u'$, which are the vertexes of subdivisions
$\Div{Div}\,\cplI$ and $\Div{Div}'\cplI$, respectively.
Process 2 terminates the execution of the protocol
with a final output $\delta((2,\{(0,\{0,1,2\}),(2,\{1,2\})\}))=1$.

\begin{figure}
  \centering
  \begin{minipage}[b]{0.31\textwidth}
    \includegraphics[scale=0.23]{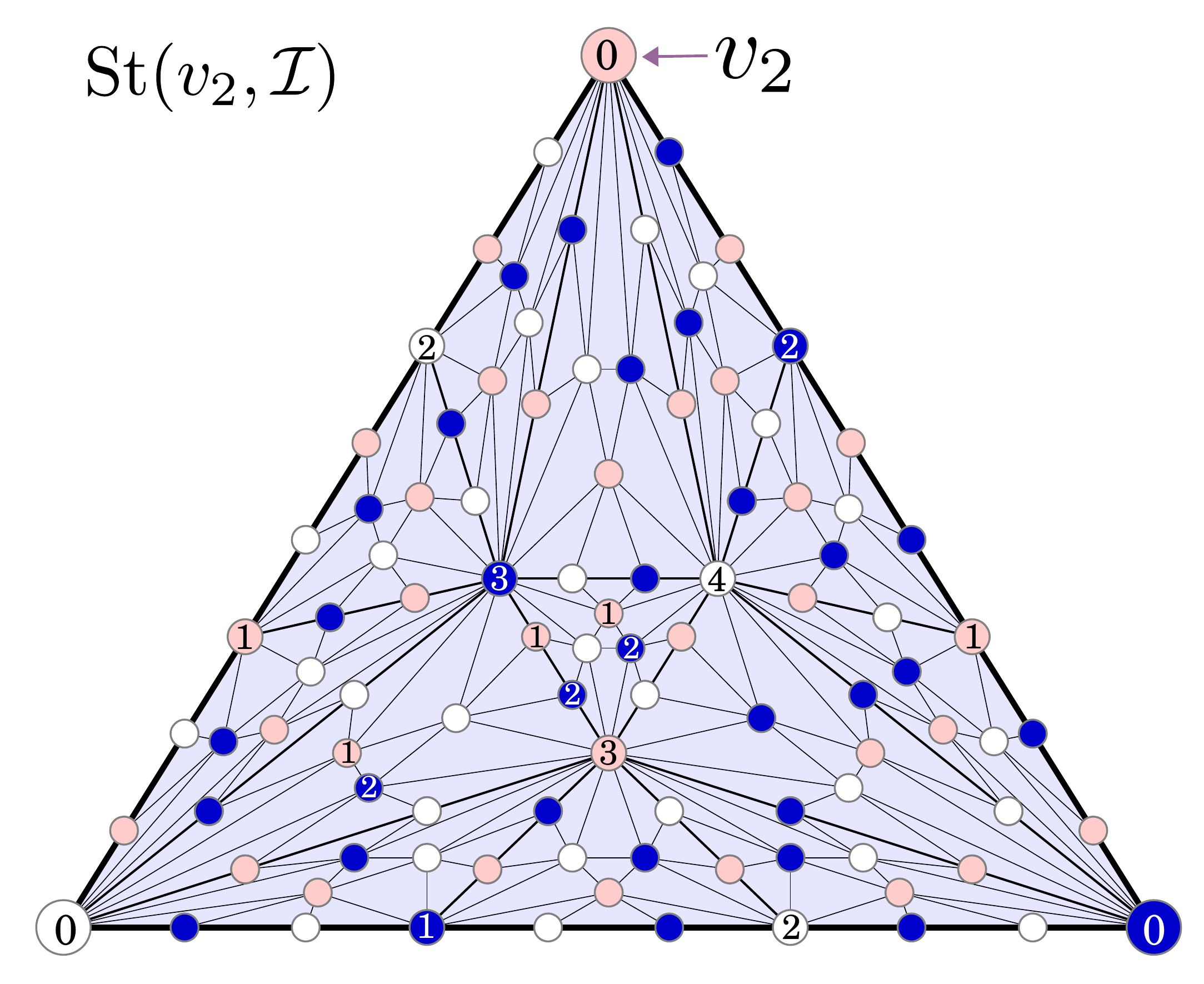}
  \end{minipage}%
  \hfil%
  \begin{minipage}[b]{0.31\textwidth}
    \includegraphics[scale=0.23]{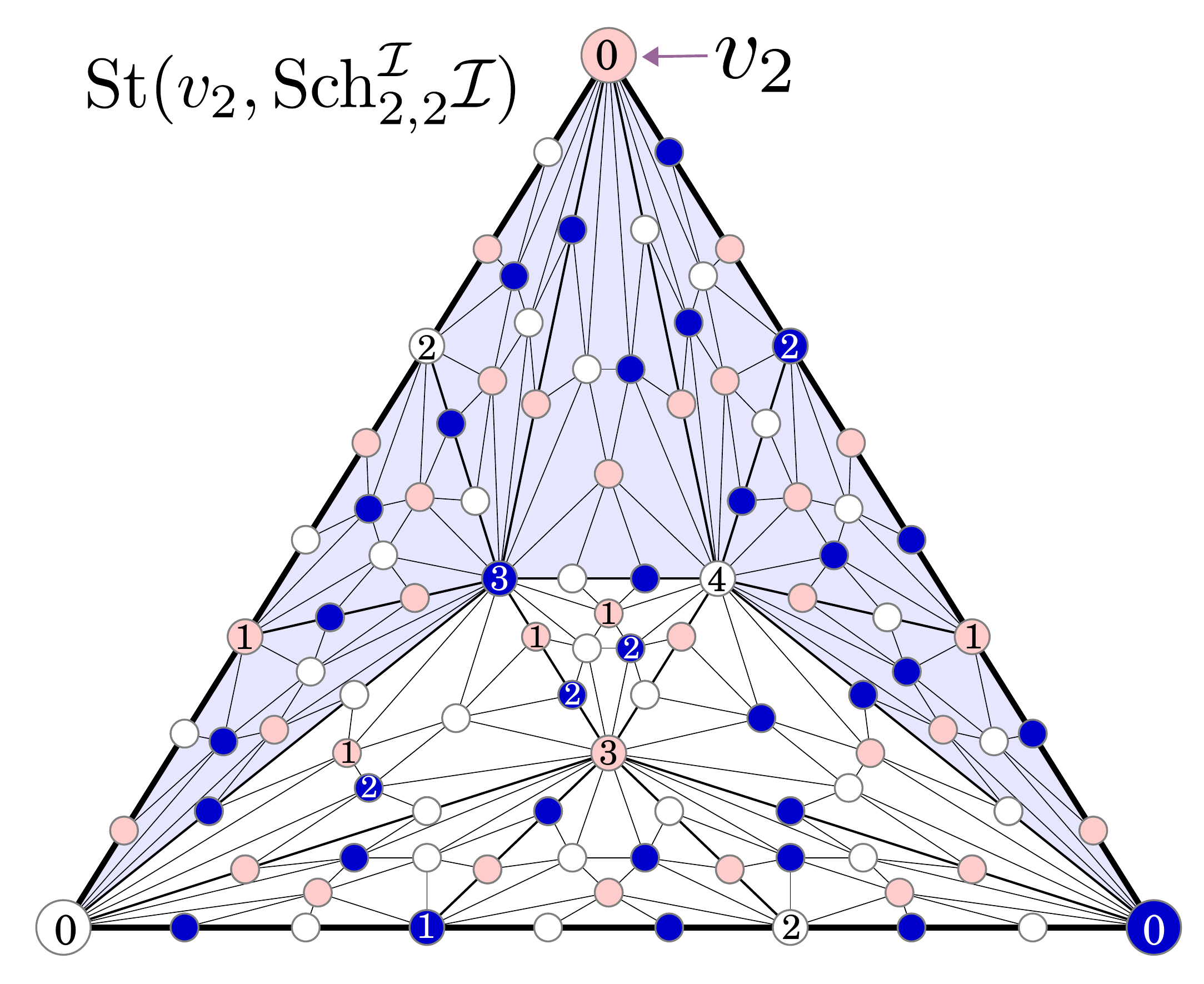}
  \end{minipage}%
  \hfil%
  \begin{minipage}[b]{0.31\textwidth}
    \includegraphics[scale=0.23]{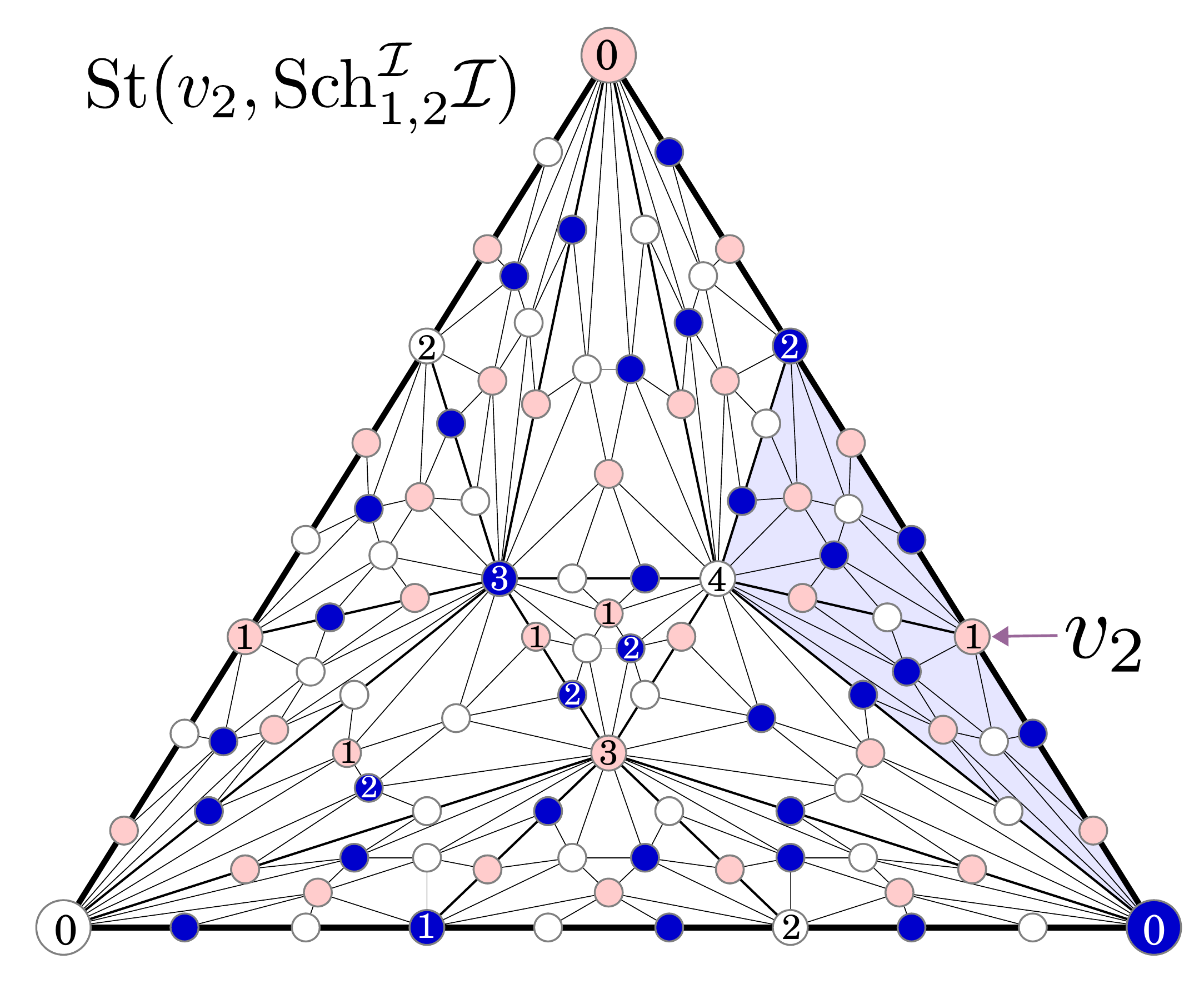}
  \end{minipage}~%

  \caption{Descendant vertexes in the iterated subdivisions
  (up to the second iteration, where process~2 is ready for decision)}
  \label{fig:descendants}
\end{figure}

This recursive call chain, however, could have decided the final
output at an earlier stage of recursion, because
the vertex $v_2\in \Div{Div}\,\cplI$ at each recursive call
can only be updated to a vertex
(of matching color) covered by $\Star(v_2,\Div{Div}\,\cplI)$
by the successive recursive calls.
Figure~\ref{fig:descendants} illustrates the first few steps,
where the shaded part indicates the simplexes of $\Chromatic^2\cplI$
covered by the star at each recursive step.
Initially, when the input complex $\cplI$ is not yet subdivided,
$\Star(v_2,\cplI)$ covers all the vertexes of $\Chromatic^2\cplI$;
After the first recursive call $\IISnap(1,2)$, $\Star(v_2,\Sch{\cplI}{2,2}\cplI)$
covers fewer vertexes in a smaller region
but they do not agree with the outputs carried by $\delta$
(they can be either $0$ or $1$);
After the second recursive call $\IISnap(1,1)$,
$\Star(v_2,\Sch{\cplI}{1,2}\cplI)$
covers even fewer vertexes and they are all carried to the
same output vertex $1$ by $\delta$.
Hence, as soon as $v_2$ is updated to the vertex
$(2,\{1,2\})$ by the recursive call $\IISnap(1,1)$,
we can decide the output of process~2 by
$\delta((2,\{1,2\}))= 1$, skipping the remaining recursive calls.

\begin{algorithm}[t]
  \caption{An optimized generic code for process $i$}
  \label{algo:IISgenopt}
  \begin{algorithmic}

    \Procedure{IIS}{$k,d$}
    \If{$\delta(\pmap{}^{-1}(v_i))=\{u\} \text{ for some } u\in V(\cplO)$} \Return $u$ \EndIf
    \State $u \gets \Call{WOScan}{k,d}$
    \If{$u\neq v_i$} $v_i \gets u$; $\Call{IIS}{k+1,n}$
    \Else~\Call{IIS}{$k,d-1$}
    \EndIf
    \EndProcedure
  \end{algorithmic}
\end{algorithm}

\begin{algorithm}[t]
  \caption{A customized code for the renaming task}
  \label{algo:IISrenam}
  \def\procN#1{\Comment{\makebox[18em][l]{\textsc{Code for process #1}}}}
  \begin{algorithmic}
    \Procedure{IIS}{$k,d$} \procN{0}
    \If{$v_0=(0,\{0,1,2\})$} \Return 4
    \ElsIf{$d=1 \wedge v_0=0$} \Return 0
    \ElsIf{$d=1$} \Return 2
    \Else
      \State $u \gets \Call{WOScan}{k,d}$
      \If{$u\neq v_0$} $v_0 \gets u$; $\Call{IIS}{k+1,n}$
      \Else~\Call{IIS}{$k,d-1$} \EndIf
    \EndIf
    \EndProcedure
  \end{algorithmic}
  \medskip

  \begin{algorithmic}
    \Procedure{IIS}{$k,d$} \procN{1}
    \If{$k=1 \wedge d=1 \wedge v_1=1$} \Return 0
    \ElsIf{$k=1 \wedge v_1=(1,\{0,1\})$} \Return 1
    \ElsIf{$k=1 \wedge v_1=(1,\{1,2\})$} \Return 2
    \ElsIf{$\left[\begin{array}[c]{l}
          k=2 \wedge v_1=(1,\tau) \text{ for some $\tau$ }\\ \text{s.t. } \{(1,\{0,1,2\}),(2,\{0,1,2\})\}\subseteq\tau
    \end{array}\right]$} \Return 2
    \ElsIf{$k=2 \wedge d=1$} \Return 3
    \Else
      \State $u \gets \Call{WOScan}{k,d}$
      \If{$u\neq v_1$} $v_1 \gets u$; $\Call{IIS}{k+1,n}$
      \Else~\Call{IIS}{$k,d-1$} \EndIf
    \EndIf
    \EndProcedure
  \end{algorithmic}
  \medskip

  \begin{algorithmic}
    \Procedure{IIS}{$k,d$} \procN{2}
    \If{$k=1 \wedge d=1 \wedge v_2=2$} \Return 0
    \ElsIf{$k=1 \wedge (v_2=(2,\{0,2\}) \vee v_2=(2,\{1,2\}))$} \Return 1
    \ElsIf{$\left[\begin{array}[c]{l}
          k=2 \wedge v_2=(2,\tau) \text{ for some $\tau$ }\\ \text{s.t. } \{(1,\{0,1,2\}),(2,\{0,1,2\})\}\subseteq\tau
    \end{array}\right]$} \Return 1
    \ElsIf{$k=2 \wedge d=1$} \Return 3
    \Else
      \State $u \gets \Call{WOScan}{k,d}$
      \If{$u\neq v_2$} $v_2 \gets u$; $\Call{IIS}{k+1,n}$
      \Else~\Call{IIS}{$k,d-1$} \EndIf
    \EndIf
    \EndProcedure
  \end{algorithmic}
\end{algorithm}

By the observation so far, we can see that the generic protocol
(Algorithm~\ref{algo:IIS}) can be further optimized
to perform fewer shared memory operations by skipping
redundant recursive calls: Once a process has reached to a point
where a sole final output is determined by the decision map $\delta$,
the remaining recursive calls can be skipped.
To make it precise,
for every $v\in \Div{Div}\,\cplI$ where $\Div{Div}\,\cplI$ is
an intermediate subdivision toward the finest subdivision
$\Chromatic^{K}\,\cplI$, let us define the \keyword{descendants} of $v$
by $\pmap{}^{-1}(v)= \{ u\in\Chromatic^{K}\,\cplI \mid \pmap{}(u)=v \}$,
where $\pmap{}: V(\Chromatic^{K}\,\cplI) \to
V(\Div{Div}\,\cplI)$ is the corresponding parent map.
We present the optimized version of the generic protocol
in Algorithm~\ref{algo:IISgenopt}.
(The omitted procedure $\WOScan$ is the same as in Algorithm~\ref{algo:IIS}.)

We notice that, for a particular decision map $\delta$ and each process $i$,
the value $\delta(\pmap{}^{-1}(v_i))$ can be precomputed
for every possible vertex assigned to $v_i$ in advance of actual
execution of the protocol.
This means that, specializing the code of Algorithm~\ref{algo:IISgenopt}
w.r.t.\ the particular decision map $\delta$,
we can generate a further optimized implementation code.

Algorithm~\ref{algo:IISrenam} gives
such a code customized for the above renaming task for 3 processes.
Observe that, precomputing descendant vertexes, program specialization
has eliminated those redundant recursive calls which are not on
reachable execution paths.

Here we notice that the optimization method discussed above
is applicable to any protocol of the form $\delta\circ\Chromatic^K\,\cplI$.
Furthermore, the optimized code
is mechanically derived by specializing the generic protocol
w.r.t.\ the concrete instance of decision map $\delta$.

Although the program specialization gives a general optimization
method, it heavily depends on each protocol instance
how much memory access can be reduced.
At one extreme, a protocol $(\cplC,\cplC, \pi\circ\Chromatic\,\cplC)$,
where $\cplC=\cls{\{v_0,\ldots,v_d\}}$ and
$\pi:V(\Chromatic\,\cplC)\to V(\cplC)$ is the parent map,
is optimized to a protocol $(\cplC,\cplC,\Phi)$
with a trivial carrier map such that $\Phi(\sigma)=\cls{\sigma}$
that performs no shared memory access.
At the other extreme,
a protocol $(\cplC,\Chromatic\,\cplC, \iota\circ \Chromatic\,\cplC)$ for
chromatic agreement task, where
$\iota:V(\Chromatic\,\cplC)\to V(\Chromatic\,\cplC)$ is an identity
vertex map, is not optimized at all by specialization,\unskip\footnote{%
  To be precise, any protocol is specialized to a code that at least skips
  the subdivisions on 0-dimensional simplexes.
  The original implementation of immediate snapshot by Borowsky and Gafni
  can also be optimized likewise.
  } %
since each vertex of $\Chromatic\,\cplC$ is assigned a different output.
Thus, there is no general theorem on the reduction in the number
or complexity of memory access.
Furthermore, the code derived by the optimization method is,
as is often the case with those obtained by automatic program generation,
inevitably less structured
than those protocols which are manually devised with
human insights (e.g., the protocols
\cite{BorowskyGafni:PODC93,GafniRajsbaum10} for renaming task).

\section{Conclusion and Future Work}
\label{sec:conclusion}

We have shown that the multi-round protocol for the
immediate snapshot by Borowsky and Gafni can be reformulated
to conform to, in terms of combinatorial topology, Kozlov's construction of
the standard chromatic subdivision via Schlegel diagrams.
This gives a topologically smoother account for the protocol,
where each round is simply a subdivision using Schlegel diagram.
This topological simplicity has led to a straightforward
method for optimizing distributed protocols defined by means of
the iterated immediate snapshot:
Each process executing the protocol narrows down
the set of possible outputs per each round
and can decide the final output at an earlier round,
beyond which the same final output is reached no matter how
the remaining rounds are executed.

The present paper exemplified that
a topologically simpler modeling can better incorporate
the theoretical results in topological studies on distributed computing
into the more practical side of distributed systems.
In this respect, it would be of interest of future investigation
to generalize the result to encompass shared memory systems
of different failure models such as
\cite{SaraphHerlihyGafni16,GafniKuznetsovManolescu14,GafniHeKuznetsovRieutord16}.
Developing an appropriate multi-round protocol that
operates on a topological model of a particular failure model,
we would be able to optimize the corresponding class of protocols.
Such an enhanced multi-round protocol would necessarily need to employ
an augmented set of memory operations in a way that
the topological structure induced from the extra operations
(e.g., the one induced from
the test-and-set operation \cite{HerlihyRajsbaum94})
is compatible with the failure model.

}

\subsection*{Acknowledgment}
I would like to thank Nayuta Yanagisawa for
his valuable comments on a draft of this paper.
This work was supported by JSPS KAKENHI Grant Number 16K00016.

{
\bibliographystyle{plainurl}
\bibliography{distrib}
}

\clearpage%
\appendix
\section*{Appendix}
\section{A Quick Look at Partial Evaluation}
\label{sec:PE}

This appendix gives a brief overview of partial evaluation,
a program optimization technique by program specialization.
Partial evaluation is a matured field that has a long history
of research. For details that cannot be covered in
the following short overview, readers are advised to consult a textbook,
say \cite{JonesGomardSestoft:PEbook}.

The fundamental idea of partial evaluation is quite simple.
Suppose we are given a program and
some of the expected inputs to it are known in advance. Then
certain portions of the program may be
\emph{precomputed} w.r.t.\ the known inputs, by which
the source program is transformed to an optimized one:
The transformed program contains
fewer computation steps to be performed at run-time.
A subpart of the program is called
\keyword{static}, if it does not depend on the inputs
to be given at run-time; Otherwise, it is called
\keyword{dynamic}. In particular, the known inputs are called
static and the remaining inputs are called dynamic.
It is the task of \emph{binding-time analysis} to identify
static parts as larger as possible for the chance of better optimization.

Let us see how a simple program that computes exponentiation
$n^m$ for non-negative integers $n$ and $m$ can be
optimized by partial evaluation.\unskip\footnote{%
  This is the typical example that first appears in introductory
  texts of partial evaluation.} %
Such a program would be simply defined in a recursive style, as follows:
\begin{align*}
  \mathrm{expt}(n,m) ~\equiv\quad &
  \textproc{if}~m=0~\textproc{then}~\textproc{return}~1~%
  \textproc{else}~\textproc{return}~n\times\mathrm{expt}(n,m-1).
\end{align*}
Suppose the second input $m$ is known $9$.
Instantiating $m$ with $9$, we get:
\begin{align*}
  \mathrm{expt}(n,\underline{9}) ~\equiv\quad &
  \textproc{if}~\underline{9=0}~\textproc{then}~\textproc{return}~\underline{1}~%
  \textproc{else}~\textproc{return}~n\times\mathrm{expt}(n,\underline{9-1}),
\end{align*}
where the underlined parts are the static ones, which are identified
by binding-time analysis. Evaluating the static subexpressions,
we obtain
\begin{align*}
  \mathrm{expt}(n,\underline{9}) ~\equiv\quad &
  \textproc{if}~\underline{\mathit{false}}~\textproc{then}~\textproc{return}~\underline{1}~%
  \textproc{else}~\textproc{return}~n\times\mathrm{expt}(n,\underline{8}).
\end{align*}
Pruning the unreachable branch and unfolding the recursive
call $\mathrm{expt}(n,\underline{8})$, we get
\begin{align*}
  \mathrm{expt}(n,\underline{9}) ~\equiv\quad &
  n\times \bigl(\textproc{if}~\underline{8=0}~\textproc{then}~\textproc{return}~\underline{1}~%
  \textproc{else}~\textproc{return}~n\times\mathrm{expt}(n,\underline{8-1}) \bigr),
\end{align*}
which reveals new static parts subject to further partial evaluation. Repeating this process,
we will obtain the final transformation result:
\begin{align*}
  \mathrm{expt}(n,\underline{9}) ~\equiv\quad &
  n\times n\times n\times n\times n\times n\times n\times n\times n.
\end{align*}

Though the above simple example of partial evaluation improves
the source program only marginally, just
removing the overhead involved
in conditional branching and recursive calls,
the effect of optimization is amplified by applying it
where execution bottleneck exists. In this
paper, we are specifically concerned with application to the shared memory bottleneck.
Another strength of partial evaluation is that the whole transformation process
is mechanizable: Once we write a simple
program, whose correctness is easier to reason about,
we may automatically obtain one that is still correct yet optimized.

Optimization by partial evaluation, or program transformation
by specialization in general,
rarely improves computational complexity.
In most cases, it improves efficiency only by a constant factor.
As for the example of exponentiation, for instance,
it is well known that an algorithm of logarithmic complexity
is obtained by the technique of repeated squaring \cite{IntroAlgo:Cormen}.
However, this kind of algorithmic leap usually needs
human insights on the mathematical structure
behind the problem to be solved.
It is a central topic
of program transformation how to optimize programs
semi-automatically ---
mostly by mechanical `calculation' on programs but
with a little exploitation of the mathematical
structure behind each particular problem.
There has been lots of work done in this direction
and several illuminating examples can be found in \cite{Bird:PearlsFunAlgo}.

\end{document}